\documentclass[nofootinbib,amsmath,amssymb,aps,pra,longbibliography,twocolumn]{revtex4-1}

\usepackage[utf8]{inputenc}
\usepackage{braket}
\usepackage{slashed}
\usepackage{subcaption}
\usepackage{graphicx}
\usepackage{amsmath}
\usepackage{amssymb}
\usepackage{amsthm}
\usepackage{marginnote}
\usepackage{xcolor}
\usepackage{hyperref}

\usepackage{xr-hyper}
\usepackage{bm}
\usepackage{gensymb}
\usepackage{mathtools}
\usepackage{lineno}
\usepackage[export]{adjustbox}
\usepackage{newtxtext,newtxmath}

\usepackage{float}
\usepackage[margin=0.5in]{geometry}
\hypersetup{
    colorlinks=true,
    linkcolor=blue,
    filecolor=magenta,      
    urlcolor=blue,
    citecolor=red,
    pdftitle={equivariant},
    pdfpagemode=FullScreen,
    }
\graphicspath{ {./images/} }

\newtheorem{lemma}{Lemma}

\newtheorem{proposition}{Proposition}

\DeclareMathOperator{\tr}{Tr}

\begin{document}
\title[]{Learning Equivariant Maps with Variational Quantum Circuits}

\author{Zachary P. Bradshaw, Ethan N. Evans, Matthew Cook, and Margarite L. LaBorde}
\affiliation{Naval Surface Warfare Center, Panama City, Florida 32407, USA}

\begin{abstract} Geometric quantum machine learning uses the symmetries inherent in data to design tailored machine learning tasks with reduced search space dimension. The field has been well-studied recently in an effort to avoid barren plateau issues while improving the accuracy of quantum machine learning models. This work explores the related problem of learning an equivariant map given two unitary representations of a finite group, which in turn allows the symmetric embedding of the data to be learned rather than simply required. Moreover, this procedure allows the learning of covariant quantum channels, which are an essential tool in quantum information theory. We demonstrate the feasibility of this task and give examples to illustrate the procedure.
\end{abstract}

\maketitle
\date{\today}

\section{Introduction}\label{sec:intro}

Quantum machine learning \cite{cao2020, grant2018, schuld2018, schuld2020circuit} has been a subject of vast interest in recent years for its feasibility in the era of noisy intermediate scale quantum (NISQ) devices \cite{preskill2018quantum}. It aims to combine the success of machine learning with the properties of quantum theory to perform learning in an efficient and accurate manner. While many avenues have been explored, a popular scheme is to make use of a variational quantum circuit in place of a neural network \cite{bishop2006,vapnik1999}. The parameters are then trained classically through one of the many available optimization routines.

A frequent issue that arises while training variational quantum algorithms is an overly large expressivity, which introduces barren plateaus into the training landscape \cite{arrasmith2022,cerezo2021,holmes2022,mcClean2018}. However, it has been shown that the natural symmetries in a given problem can be harnessed to alleviate this issue \cite{larocca2022,meyer2023,nguyen2022}, and geometric quantum machine learning, as discussed at length in \cite{ragone2023}, uses symmetries in this manner. In this framework, the data embedding is represented as an equivariant map between different representations. By building the circuit with only gates which respect the symmetry of the problem, we can be guaranteed that the resulting estimate completely respects this symmetry, thereby reducing the size of the space through which the model must search for a solution.  This process mimics the physics of Lagrangian mechanics, where symmetries in the evolution of a system simplify the corresponding equations of motion. In the same manner, geometric quantum machine learning can reduce the size of the solution space in QML \cite{larocca2022,meyer2023,nguyen2022}. Such symmetry-preserving circuits have previously proven useful in developing quantum eigensolvers \cite{SSY20,Gard2020,BGAMBE21,LXYB22}.

In this work, we show that the symmetric variational quantum circuit construction can be repurposed to search for equivariant maps between representations. We review the construction of symmetry-respecting variational quantum circuits and show how these methods can be inverted to learn an equivariant map between two given representations, assuming such a map exists. Examples are then provided to illustrate this methodology. 

As representation theory is a natural language for describing the symmetries inherent in a physical system, the methods presented here are of interest to physicists from a variety of subdisciplines. In particular, equivariant maps often take the form of covariant quantum channels \cite{holevo2002}. A quantum channel $\mathcal{N}:\mathcal{H}_A \to \mathcal{H}_B$ mapping between quantum states is said to be covariant with respect to a group $G$ with representations $W_A$ and $W_B$ if 
\begin{equation}\label{eq:covariant}
\mathcal{N}_{A\rightarrow B}\circ\mathcal{W}_{A}(g)=\mathcal{W}_{B}(g) \circ \mathcal{N}_{A\rightarrow B} \, ,
\end{equation}
for all $g\in G$, where $\mathcal{W}_A(g)\rho=W_A(g)\rho W_A(g)^\dagger$. Thus, our results facilitate the learning of covariant quantum channels \cite{CDP09,cemin2024machinelearningquantumchannels}. Physically, these channels describe the symmetric dynamics of a system. For Hamiltonian dynamics, these covariant symmetries can represent forbidden state transitions \cite{vojta2003quantum} or superselection rules \cite{Wick1952,Aharonov1967superselection}. In resource theory, covariant channels cannot increase asymmetry measures --- a fact which can be thought of as an extension of Noether's theorem for quantum states \cite{marvian2014Noether}. 

Bijective linear equivariant maps, sometimes referred to as intertwining maps, also form a natural notion of equivalence between representations \cite{hall2003,steinberg2011}, and the classification of the representations of a group according to this equivalence relation are of interest to mathematicians working in a variety of subdisciplines such as geometry, statistics, and representation theory itself. Thus, the ability to learn an intertwining map between equivalent representations is desirable, and the ability to give a measure of how equivariant a map is will be of interest to those studying maps which respect only a substructure between representations. 

Deep learning theorists also rely on equivariant maps to utilize data symmetries for improved performance even in classical machine learning \cite{mondal2023advances,ouderaa2023learning, minartz2023neuralsimulators,wang2022approximately}. Graph neural networks, for instance, must respect permutation symmetry, making learning equivariance an important task in that framework \cite{huang2023approx}. The task of learning equivariant maps for classical deep neural networks has been recently investigated in \cite{ouderaa2023learning,benton2020learning}, and naturally arises when considering symmetry-constrained dynamics \cite{wang2022approximately}.

The rest of this work is structured as follows: In Section~\ref{sec:sym}, we review the construction of symmetry-respecting variational quantum circuits, and we complement this theoretical discussion with a concrete example in Section~\ref{sec:class}, where we solve the toy problem of classifying vertical and horizontal lines in $2\times 2$ images with a single parameter in a variational quantum machine learning algorithm that respects the geometric structure of this problem. In Section~\ref{sec:equivariant}, we show how the methods of the previous section can be inverted to learn an equivariant map between two given representations, assuming such a map exists. Several examples are then provided to illustrate this methodology. In Section~\ref{sec:intertwine}, we show that this method can be extended to learning intertwining maps, despite the data embedding being nonlinear. Finally, we give concluding remarks in Section~\ref{sec:conclusion}.

\section{Symmetry Assisted Quantum Machine Learning}\label{sec:sym}

We begin by reviewing the mathematical foundations of this problem, restricting our attention to discrete symmetries for the scope of this work. A unitary representation $W:G\to \mathcal{U}(\mathcal{H})$ of a finite group $G$ is a group homomorphism of $G$ into the group of unitary matrices acting on the Hilbert space $\mathcal{H}$. The unitary condition is assumed so that quantum states are mapped to quantum states. When $W$ is injective, we call the representation faithful, as its image is isomorphic to $G$ by the first isomorphism theorem \cite[Chapter~3]{dummit2004}. Such a representation can therefore be thought of as a translation of the description of the group from its abstract group language into the language of quantum gates. Representations that are not faithful respect only a substructure of the original group $G$ in the sense that the image under $W$ is isomorphic to a proper subgroup of $G$.

We aim to apply representation theory to the study of classification problems. Given a data set $\mathcal{S}=\{(x,y)\}$ consisting of pairs of data points $x$ drawn from some data space $\mathcal{R}$ with labels $y$ in some set $\mathcal{L}$, we assume there is a function $f:\mathcal{R}\to\mathcal{L}$ which correctly maps a data point to its associated label, i.e. $f(x)=y$. We will make the assumption that $\mathcal{R}$ is endowed with a vector space structure. Typically, we approximate $f$ using a variational quantum circuit (VQC). Of course, this requires us to embed the data into a Hilbert space \cite{lloyd2020quantum}, and this alone is a nontrivial piece of the VQC, as the choice of embedding can have profound consequences on the trainability of the model (see \cite{laRose2020} for robustness results). For now, we will assume that an embedding $\mathcal{E}:\mathcal{R}\to\mathcal{U}(\mathcal{H})$ has been chosen, which by convention we always take to act on the initial all-zeroes state, $\ket{\vec{0}}$, which could in principle consist of multiple qubits.

The most common form for a VQC in the literature is given by an embedding followed by a parameterized unitary and a measurement. For this reason, we restrict ourselves to VQCs which produce an estimate of the form
\begin{equation}\label{eq:qml_eqn}
f_{\phi}(x)=\tr\left[U_\phi(\mathcal{E}(x)\ket{\vec{0}}\!\!\bra{\vec{0}}\mathcal{E}(x)^\dagger) U^\dagger_\phi\mathcal{O}\right],
\end{equation}
where $\mathcal{O}$ is some chosen observable, and $U_\phi$ is the parametrized unitary at the heart of the VQC. In QML, $f_{\phi}$ is then used to model the behavior of $f$ as closely as possible by ensuring that $\Vert f_\phi (x) - y\Vert$ is minimal for every data-label pair $(x,y)$.

Suppose now that the data respects some symmetry described by a group $G$ with representation $V:G\to GL(\mathcal{R})$, where $GL(\mathcal{R})$ is the general linear group on $\mathcal{R}$.
The target function $f$ is called $G$-invariant if
\begin{equation} 
f(V(g)x)=f(x)
\end{equation}
for all $g\in G$ and for all $x\in\mathcal{R}$. This phenomenon appears, for example, in many image classification problems where the label associated to the image does not change if the image is rotated. 

Since we are trying to estimate $f$ with $f_{\phi}$, it makes sense to introduce an inductive bias into the VQC model by demanding $G$-invariance. One way to accomplish this is by adding a penalty term to the loss function when the optimization routine is performed which penalizes the output when it lacks $G$-invariance. This property can be efficiently tested for on a quantum computer as in \cite{laBorde2023}; however, this often extends the training time of the model. An alternative approach is to build the symmetry directly into the VQC, and to do this, we need a notion of equivariance.

The representations $V$ and $W$ are called equivariant if there exists a map $T:\mathcal{R}\to\mathcal{H}$ (called an equivariant map) such that
\begin{equation}\label{eq:equivariance-condition} T\circ V(g)=W(g)\circ T,
\end{equation}
which bears a striking resemblance to the covariance condition given in \eqref{eq:covariant}.
When $T$ is bijective and linear, we call it an isomorphism of the representations, and the representations themselves are called equivalent. If we are given the representation $V:G\to GL(\mathcal{R})$ and a map $T:\mathcal{R}\to\mathcal{H}$, then a representation $W:G\to\mathcal{U}(\mathcal{H})$ is said to be induced by $T$ whenever $T$ is equivariant. 

Additionally, an operator $A$ is called $G$-invariant when
\begin{equation} [W(g),A]=0
\end{equation}
for all $g\in G$. That is, $A$ belongs to the commutant of $W(G)$. This fact is important in the coming derivations, as we enforce symmetry directly in the VQC by ensuring $G$-invariance throughout the circuit.

In the context of a VQC, we are given a problem with some symmetry described by a representation $V$. If we are also given an embedding $\mathcal{E}$, then a representation $W$ induced by $\mathcal{E}$ encodes the symmetry of the problem in the Hilbert space picture. In other words, the embedding carries the symmetry structure from the data space to the Hilbert space so that this structure can be respected by the VQC. Strictly speaking, the embedding is a map $\mathcal{E}:\mathcal{R}\to\mathcal{U}(\mathcal{H})$ which by convention acts on the all-zeroes state $\ket{\vec{0}}$, thus the map $T$ in \eqref{eq:equivariance-condition} is really given by $T(x)=\mathcal{E}(x)\ket{\vec{0}}$. When we say that $W$ is induced by the embedding, what we really mean is that $W$ is induced by the map $T:\mathcal{R}\to\mathcal{H}$, which is defined by the embedding. If instead we are given two representations and we want to guarantee that the embedding is an equivariant map between them, then this is a non-trivial task which we take up in Section \ref{sec:equivariant}.

\begin{proposition} \label{prop:symmetric} Let $V$ be a representation of $G$ on the data space $\mathcal{R}$ and let $\mathcal{E}$ be an embedding of the data space into $\mathcal{H}$. Let $W$ be a representation induced by the embedding. Then the estimate $f_\phi$ defined by \eqref{eq:qml_eqn} is $G$-invariant if the parameterized unitary $U_\phi$ and the observable $\mathcal{O}$ are $G$-invariant. That is, 
\begin{equation}f_\phi(V(g)x)=f_\phi(x)
\end{equation}
for all $x\in\mathcal{R}$ and $g\in G$ whenever $U_\phi$ and $\mathcal{O}$ are $G$-invariant with respect to $W$.
\end{proposition}
\begin{proof} Since $W$ is induced by $\mathcal{E}$, we have
\begin{align} f_\phi(V(g)x)
&=\tr\left[U_\phi\mathcal{E}(V(g)x)\ket{0}\!\!\bra{0}\mathcal{E}(V(g)x)^\dagger U^\dagger_\phi\mathcal{O}\right]\\
&=\tr\left[U_\phi W(g)\mathcal{E}(x)\ket{0}\!\!\bra{0}\mathcal{E}(x)^\dagger W(g)^\dagger U^\dagger_\phi\mathcal{O}\right].
\end{align}
Now applying the $G$-invariance of $U_\phi$ produces
\begin{align}
&=\tr\left[W(g)U_\phi \mathcal{E}(x)\ket{0}\!\!\bra{0}(\mathcal{E}(x))^\dagger  U^\dagger_\phi (W(g))^\dagger\mathcal{O}\right],
\end{align}
and the cyclicity of the trace, the unitarity of $W$, and the $G$-invariance of $\mathcal{O}$ together yield
\begin{align}
&=\tr\left[U_\phi \mathcal{E}(x)\ket{0}\!\!\bra{0}(\mathcal{E}(x))^\dagger  U^\dagger_\phi (W(g))^\dagger\mathcal{O}W(g)\right]\\
&=\tr\left[U_\phi \mathcal{E}(x)\ket{0}\!\!\bra{0}(\mathcal{E}(x))^\dagger  U^\dagger_\phi \mathcal{O}\right]\\
&=f_\phi(x).
\end{align}
This completes the proof.
\end{proof}
Through Proposition~\ref{prop:symmetric}, we see that we may enforce $G$-invariance in $f_\phi$ by using only gates and observables which commute with the representation induced by the embedding, but how do we come up with a satisfactory set of such operations? The following lemma is helpful.

\begin{lemma}\label{lemma:commutator} If $[X,Y]=0$, then $[e^{i\phi X},Y]=0$ for all $\phi\in\mathbb{R}$.
\end{lemma}
\begin{proof} The proof is a simple computation using the Taylor series for the exponential. We have
\begin{align} [e^{i\phi X},Y]&=\left[\sum_{n=0}^\infty \frac{i^n\phi^n X^n}{n!},Y\right]=\sum_{n=0}^\infty \frac{i^n\phi^n}{n!}[X^n,Y].
\end{align}
Since $[X^n,Y]=X^{n-1}[X,Y]+[X^{n-1},Y]X$, proceeding by induction will produce a sum of terms involving at least one power of $[X,Y]$, which vanishes by assumption. Thus, the last equality vanishes.
\end{proof}

By Stone's theorem \cite{hall2013quantum}, for every strongly continuous one-parameter unitary group $U_\phi$, there is a self-adjoint operator $H$ such that $U_\phi=e^{i\phi H}$. This is why the lemma is so useful; it tells us that we can find parametrized gates that commute with the representation induced by the embedding by checking that the infinitesimal generator commutes with this representation. For example, if $[Z,W(g)]=0$ for all $g\in G$, then every rotation gate along the $Z$-axis is $G$-invariant.

Given a set of gates, we can produce a set of $G$-invariant gates by projecting onto the set of $G$-invariant operators, which is accomplished by a twirling operation.

\begin{lemma}\label{lemma:twirl} Let $X$ be a self-adjoint operator and let $W$ be a representation of $G$ as before. Then the twirl of $X$
\begin{equation} \mathcal{T}_G(X)=\frac{1}{|G|}\sum_{g\in G}W(g)XW(g)^\dagger
\end{equation}
is self-adjoint and $G$-invariant.
\end{lemma}
\begin{proof} To see that $\mathcal{T}_G(X)$ is self-adjoint, observe that
\begin{align} (\mathcal{T}_G(X))^\dagger&=\frac{1}{|G|}\sum_{g\in G} \left(W(g) X W(g)^\dagger\right)^\dagger\\
&=\frac{1}{|G|}\sum_{g\in G}W(g) X^\dagger W(g)^\dagger\\
&=\frac{1}{|G|}\sum_{g\in G}W(g) X W(g)^\dagger=\mathcal{T}_G(X).
\end{align}
To see that $\mathcal{T}_G(X)$ is $G$-invariant, let $g\in G$ and notice that
\begin{align} W(g)\mathcal{T}_G(X)&=\frac{1}{|G|}W(g)\sum_{g'\in G} W(g')XW(g')^\dagger\\
&=\frac{1}{|G|}\sum_{g'\in G} W(gg')XW((g')^{-1}).
\end{align}
Re-indexing with $h=gg'$ produces
\begin{align} W(g)\mathcal{T}_G(X)&=\frac{1}{|G|}\sum_{h\in G} W(h)XW(h^{-1}g)\\
&=\frac{1}{|G|}\sum_{h\in G}W(h)XW(h)^\dagger W(g)\\
&=\mathcal{T}_G(X)W(g),
\end{align}
and this completes the proof.
\end{proof} 

This lemma gives a technique for constructing VQCs which respect the symmetry of a problem. First find a representation $W$ induced by the embedding $\mathcal{E}$. Then construct a $G$-invariant gate set by twirling self-adjoint operators and using them as the infinitesimal generators of the gate set. Finally, choose an observable which is $G$-invariant. By Proposition~\ref{prop:symmetric}, the estimate produced by the model will now be $G$-invariant.  Given any three of these building blocks ---the embedding, the initial group representation, the embedded representation, and the observable measured--- the remaining piece could in principle be learned or determined to suit the task at hand, and we will return to this point in Section~\ref{sec:equivariant}.

\subsection{Example: Classifying Lines in Images}\label{sec:class}
Let us demonstrate how to construct VQCs which respect the symmetry of a problem through a simple example. Consider the problem of classifying vertical and horizontal lines in $2\times2$ grayscale images (see Figure~\ref{fig:samps}) which are easily generated in Python.

\begin{figure}[h]
\centering
\includegraphics[
width=1.5in
]{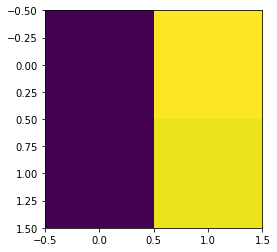}
\includegraphics[
width=1.5in
]{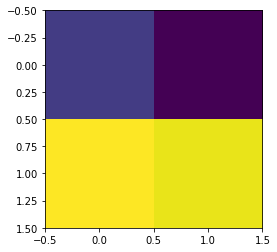}
\caption{Data samples with vertical line (left) and horizontal line (right). We are aiming to detect the lines formed by the yellow pixels. Note that the pixels are actually grayscale, and the coloring is added to improve visualization.}
\label{fig:samps}
\end{figure}

An image of a vertical line is still an image of a vertical line when it is rotated by $180\degree$. Similarly, the orientation of a horizontal line is preserved by this rotation. In fact, we also preserve the same labeling when the image is flipped about a vertical line through its midpoint. Each of these operations has order two since they undo themselves; thus, this image classification problem has symmetry group $G=C_2\times C_2$, where $C_2$ is the cyclic group of order two. 

There is a one-to-one correspondence between the $2\times2$ matrix of pixel values in the image and the vector obtained by flattening this matrix, and we take the collection of the latter as our data space. Then a rotation by 180 degrees is accomplished by the matrix $X\otimes X$ and a reflection is accomplished by the matrix $I\otimes X$ where $X$ is the Pauli-$X$ gate. Indeed, we have
\begin{equation}
\begin{pmatrix} 0&0&0&1\\0&0&1&0\\0&1&0&0\\1&0&0&0
\end{pmatrix}
\begin{pmatrix} a\\b\\c\\d
\end{pmatrix}
=\begin{pmatrix}d\\c\\b\\a
\end{pmatrix},
\end{equation}
and mapping back to the $2\times2$ matrix of pixel values shows that the image has been rotated by 180 degrees. Similarly,
\begin{equation}
\begin{pmatrix} 0&1&0&0\\1&0&0&0\\0&0&0&1\\0&0&1&0
\end{pmatrix}
\begin{pmatrix} a\\b\\c\\d
\end{pmatrix}
=\begin{pmatrix}b\\a\\d\\c
\end{pmatrix},
\end{equation}
and mapping back to the $2\times2$ matrix of pixel values shows that the image has been reflected about the vertical axis passing through the midpoint. Thus, the $C_2\times C_2$ symmetry is accomplished by the representation $V$ mapping $(1,0)\to X\otimes X$ and $(0,1)\to I\otimes X$, where $1$ is the generator of $C_2$, the rotation is identified with $(1,0)$, and the flip is identified with $(0,1)$. It is not hard to see that this representation is faithful; it respects the full structure of the group. Note that a flip about the \textit{horizontal} axis passing through the midpoint is also a symmetry, but this is not an independent symmetry, as it can be built from the rotation and flip operations; in fact, this horizontal flip corresponds to the element $(1,1)$ in $C_2\times C_2$ which is represented by $X\otimes I$. 

The simplest nontrivial embedding for this problem is the amplitude embedding, wherein a data point $(a,b,c,d)^T$ is mapped to the quantum state with the same coefficients (up to normalization) in the computational basis:
\begin{equation} (a,b,c,d)^T\to\frac{a\ket{00}+b\ket{01}+c\ket{10}+d\ket{11}}{\sqrt{a^2+b^2+c^2+d^2}}.
\end{equation}
With this embedding, the induced representation is the same as the representation $V$, so that $G$-invariant operators need only commute with $I\otimes X$ and $X\otimes X$. Equivalently, we could take the vertical and horizontal flips to be the generators of $C_2\times C_2$, so that a $G$-invariant operator need only commute with $I\otimes X$ and $X\otimes I$. Clearly the generators themselves satisfy this property since $(I\otimes X)(X\otimes I)=X\otimes X=(X\otimes I)(I\otimes X)$. Thus, the rotations by an arbitrary angle about these axes are also $G$-invariant by Lemma~\ref{lemma:commutator}, and so we can use them to parametrize our circuit. These rotations correspond to $I\otimes R_X(\theta)$ and $R_X(\theta)\otimes I$. In fact, by invoking our commutation relations, it is easily shown that the only two-qubit $G$-invariant operators have the form $aI\otimes I+bI\otimes X+cX\otimes I+dX\otimes X$ for some constants $a,b,c,d$. Since each of the terms commutes with the rest, the exponentiation splits into a product of four exponentials by the Baker-Campbell-Hausdorff formula. That is, the exponential $\exp(i\phi(aI\otimes I+bI\otimes X+cX\otimes I+dX\otimes X))$ is equivalently given by
\begin{align}
    e^{i\phi aI\otimes I}e^{i\phi bI\otimes X}e^{i\phi cX\otimes I}e^{i\phi dX\otimes X}.
\end{align} Each of the remaining exponentials is a tensor product of local operations (specifically, identity operations and $X$-rotations). Thus, there are no entanglement generating operations which preserve the symmetry of the problem, and the gateset is not universal. 

Next, we will adjust this embedding slightly by introducing a readout qubit. The embedding can now be thought of as the tensor product of the amplitude embedding with an identity mapping on the readout qubit, which is initialized to $\ket{0}$. The induced representation $W$ is now defined by
\begin{equation} \left(\mathcal{E}_{amp}\otimes I\right)(V(g)x)\ket{\vec{0}}=W(g)\left(\mathcal{E}_{amp}\otimes I\right)(x)\ket{\vec{0}},
\end{equation}
so that $V$ and $W$ are again identical, although this time they act on three qubits. We must find self-adjoint operations that commute with $X\otimes I\otimes I$ and $I\otimes X\otimes I$. Clearly $I\otimes I\otimes H$ satisfies this property for any single qubit self-adjoint operator $H$. In particular, we see that we may make our measurement with respect to $Z$ on the readout qubit at the end of the circuit while still satisfying the hypothesis of Proposition~\ref{prop:symmetric}. The $X$ gate commutes with these operations, no matter which qubit it acts on, so we are free to use $R_x(\theta)$ gates in the construction of the variational circuit. The self-adjoint operation $I\otimes I\otimes Z+I\otimes X\otimes X$ also commutes with both $X\otimes I\otimes I$ and $I\otimes X\otimes I$. Moreover, $[Z,X]\ne0$, so that the exponentiation of $I\otimes I\otimes Z+I\otimes X\otimes X$ is not a product of local operations and can produce entanglement between a data qubit and the readout qubit. A similar conclusion can be made for the operation $I\otimes I\otimes Z+X\otimes I\otimes X$. The circuit is therefore made symmetric by measuring the readout qubit with respect to $Z$ (or any other single qubit unitary) and building the variational circuit from the gate set $\{R_X(\theta), e^{i\phi\left(I\otimes I\otimes Z+I\otimes X\otimes X\right)},e^{i\phi\left(I\otimes I\otimes Z+X\otimes I\otimes X\right)}\}$.

An example circuit for solving the horizontal/vertical line classification problem with a built in symmetry bias is shown in Figure~\ref{fig:2x2}. This circuit has a single trainable parameter and solves the problem with 100\% accuracy after a standard training procedure.

\begin{figure}[h]
\centering
\includegraphics[
width=\columnwidth
]{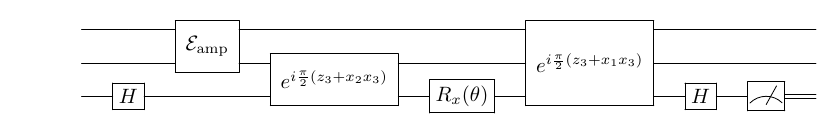}
\caption{Example of symmetric circuit for solving vertical/horizontal line classification problem. For brevity, we have used the notation $z_i$ to mean the $Z$ operator acting on the $i$-th component of the tensor product space.}
\label{fig:2x2}
\end{figure}

\section{Learning Equivariant Maps}\label{sec:equivariant}
 In the previous section, we reviewed the use of symmetry in VQCs to reduce the size of the parameter search space. In this section, we tackle the separate problem of learning an equivariant map between two given representations. We do so by leveraging the symmetric VQC construction developed above. Our aim is to build this equivariant map by forcing the estimate, $h_\theta$, produced by a new VQC to respect the same $G$-invariance as $f$. This procedure allows us to learn the desired embedding which can then be used for any classification tasks desired. The results that follow focus on this goal.

Given the representation $V:G\to GL(\mathcal{R})$ describing the symmetry respected by the data space and another representation $W:G\to \mathcal{U}(\mathcal{H})$, we can learn an equivariant map between them, assuming one exists. In the context of the symmetry-assisted VQC outlined above, this equivariant map appears as the data embedding, so that we are learning a data embedding which properly respects the symmetry of the data. We begin by embedding the data with a parameterized unitary $\mathcal{E}_\theta(x)$ and follow this with a unitary which is $G$-invariant with respect to $W$. The expected value of the circuit will then be $G$-invariant with respect to $V$ whenever the parameterized unitary is an equivariant map between $V$ and $W$ and the observable is $G$-invariant with respect to the $W$ representation. To train the equivariance, we sample from both the data space and the group and minimize the loss function 
\begin{equation}\label{eq:loss}
  L_g(x)=( h_\theta(V(g)x)-h_\theta(x)) ^2 \, .  
\end{equation}

Thus, learning an equivariant map amounts to performing a regression minimizing $L_g(x)$ for any pair $(g,x)\in G\times \mathcal{R}$. In fact, it suffices to sample over the generators of $G$ since invariance under the action of the generators implies invariance under the action of any other group element. It should be noted that the value of $L_g(x)$ can be interpreted as a measure of how equivariant the embedding $\mathcal{E}_\theta$ is; that is, how well it transfers the notion of symmetry from one representation to another. In most cases, $L_g(x)$ is not actually zero, and so we are really learning maps which are in a sense \textit{almost} equivariant.  By minimizing this loss, we can achieve an embedding that respects the symmetry of our problem. This learned embedding can now be used to construct a symmetric VQC for the purposes of approximating the target function $f$ as discussed in Section~\ref{sec:sym}.

In the experiments that follow, we train to 150 epochs with a batch size of 100 and a learning rate of 0.1, and we sample only from a generating set for each group. We simulate our quantum circuits using PennyLane \cite{bergholm2018pennylane} and optimize using vanilla gradient descent with the gradient computed by PennyLane's built-in SPSA \cite{spall1998overview} method, which shifts all parameters of the VQC simultaneously and then approximates the gradient using these shifts and a finite-difference method. The data, which is shuffled at every epoch, is generated in Python.

\subsection*{Example: $C_2$ Symmetry}

\begin{figure}[h!]
\centering
\includegraphics[
width=3.5in]{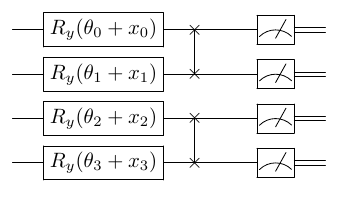}
\caption{Circuit for learning equivariance between two representations of $C_2$. Each rotation gate ($R_y$) parameterized by an angle $\theta_i$ that is then learned.}
\label{fig:c2_circ}
\end{figure}

Consider again the problem of classifying vertical and horizontal lines in a $2\times2$ image and recall that the classification of a line is invariant under a flip across the vertical axis. Since this flip is its own inverse operation, the group it generates is the cyclic group of order two. Thus, we have a representation $V$ of $C_2$ in the dataspace. On the other hand, the operation $\textnormal{SWAP}_{01}\textnormal{SWAP}_{23}$ generates a representation $W$ of $C_2$ in the Hilbert space setting. We can attempt to find an equivariant map between these representations by introducing a parameterized embedding $\mathcal{E}_\theta(x)=R_Y(\theta_0+x_0)\otimes R_Y(\theta_1+x_1)\otimes R_Y(\theta_2+x_2)\otimes R_Y(\theta_3+x_3)$, where $\theta=(\theta_0,\theta_1,\theta_2,\theta_3)$ and $x=(x_0,x_1,x_2,x_3)$. 

Note that the architecture of the embedding is a free parameter; we have a choice of what gates to allow in its construction. Here, we choose a simple tensor product of rotation gates to easily showcase our procedure. Let $g$ be the generator of $C_2$ and observe that $ \mathcal{E}_\theta\circ V(g)(x) $ is given by
\begin{equation}
R_Y(\theta_0+x_1)\otimes R_Y(\theta_1+x_0)\otimes R_Y(\theta_2+x_3)\otimes R_Y(\theta_3+x_2).
\end{equation}
On the other hand, $W(g)\circ \mathcal{E}_\theta(x)$ is given by 
\begin{equation} R_Y(\theta_1+x_1)\otimes R_Y(\theta_0+x_0)\otimes  R_Y(\theta_3+x_3)\otimes R_Y(\theta_2+x_2).
\end{equation}

\begin{figure}
    \centering
    \includegraphics[width=3.5in]{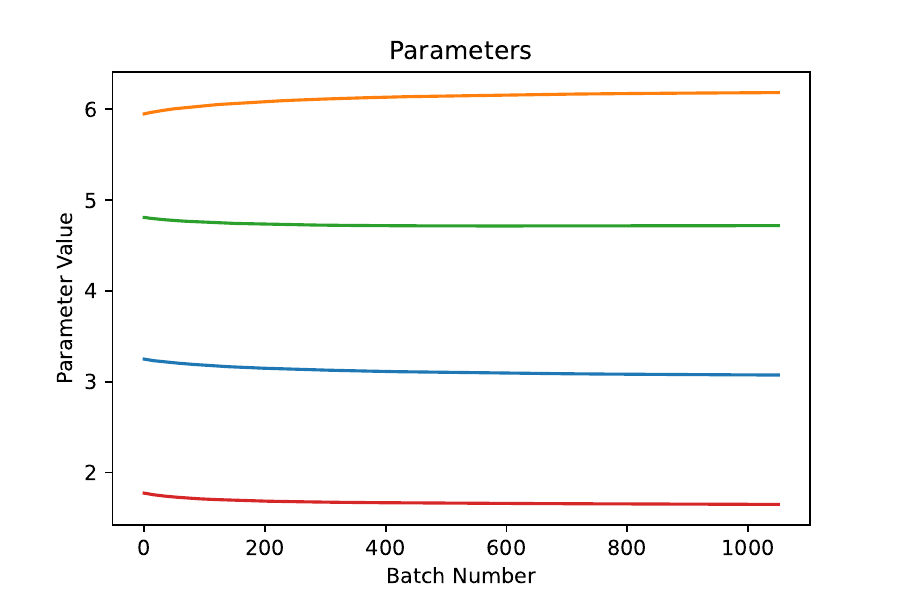}
    \caption{Parameter curves for $C_2$ symmetry}
    \label{fig:c2params}
\end{figure}

\begin{figure}
    \centering
    \includegraphics[width=3.5in]{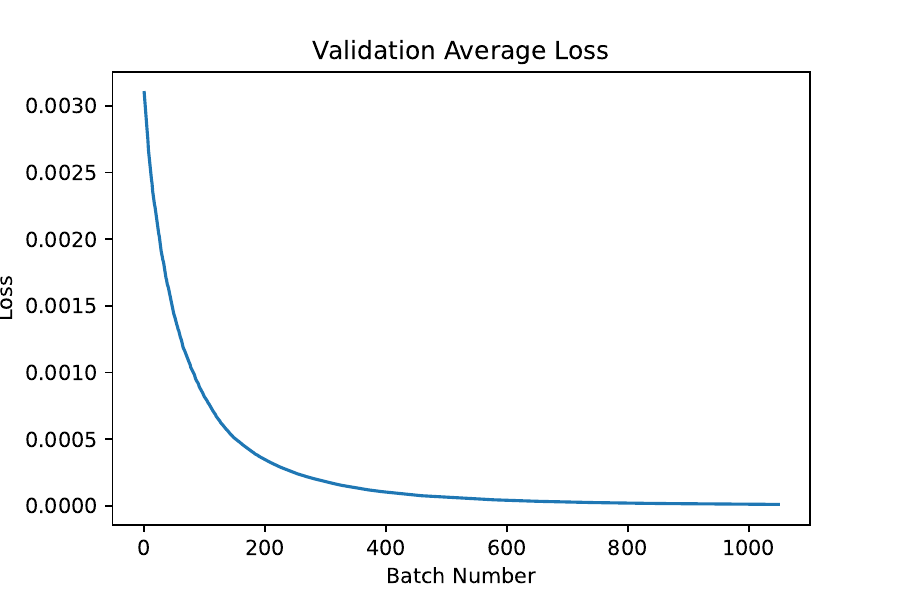}
    \caption{Average loss for $C_2$ symmetry.}
    \label{fig:c2}
\end{figure}

We see that the periodicity of the rotation operators allows for several choices of the $\theta_i$ which produce equivariance. For example, when $\theta_0=\theta_1+2\pi n$ and $\theta_2=\theta_3+2\pi m$ for some $m,n\in\mathbb{Z}$, it follows that $\mathcal{E}_\theta$ is an equivariant map between the representations. Alternatively, if each of the pairs $\theta_0,\theta_1$ and $\theta_2,\theta_3$ differs by an odd multiple of $\pi$, the resulting signs cancel and we again achieve equivariance. 

In our simulation we measure with respect to the observable $X\otimes X\otimes Z\otimes Z$, which commutes with the operation  $\textnormal{SWAP}_{01}\textnormal{SWAP}_{23}$ and is therefore $G$-invariant, and find that the model learns the parameters $\theta_0 \approx 3.07\approx\pi$, $\theta_1\approx6.18\approx 2\pi$,  $\theta_2\approx4.71\approx 3\pi/2$, and $\theta_3\approx1.65\approx\pi/2$, which aligns with our expectation. The circuit is shown in Figure~\ref{fig:c2_circ}, with the parameter curves and a plot of the average validation loss shown in Figure~\ref{fig:c2params} and Figure~\ref{fig:c2} respectively.

\subsection*{Example: $C_2\times C_2$ Symmetry} 

Let us add an independent symmetry to the mix. The flip $F_h$ about the horizontal axis in the $2\times2$ line classification problem is independent from the vertical flip $F_v$, and we will now endeavor to take both into account. The group generated by these independent flips is a direct product of cyclic groups $C_2\times C_2$. We will make use of the same embedding as for the $C_2$ symmetry: $\mathcal{E}_\theta(x)=R_Y(\theta_0+x_0)\otimes R_Y(\theta_1+x_1)\otimes R_Y(\theta_2+x_2)\otimes R_Y(\theta_3+x_3)$, but this time we expect the parameters to train differently. Indeed, the representation of $C_2\times C_2$ acting on our Hilbert space will be given by $W(F_v)=\textnormal{SWAP}_{01}\textnormal{SWAP}_{23}$ and $W(F_h)=\textnormal{SWAP}_{02}\textnormal{SWAP}_{13}$, so that the unitary applied after the embedding must commute with these two operators. Of course, any choice of a single qubit gate applied simultaneously to all qubits will do the trick. Thus, we take our $C_2\times C_2$-invariant unitary to be $U=Y^{\otimes4}$ and we take our observable to be $Z^{\otimes4}$.

\begin{figure}[h!]
\centering
\includegraphics[
width=3.5in,right
]{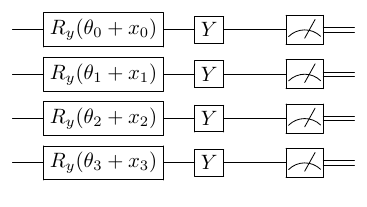}
\caption{Circuit for learning equivariance between two representations of $C_2\times C_2$. Rotation gates $R_y(\theta_i)$ are randomly initialized and then updated throughout the learning process.}
\label{fig:c2c2_circ}
\end{figure}


\begin{figure}[H]
\centering
\includegraphics[
width=3.5in
]{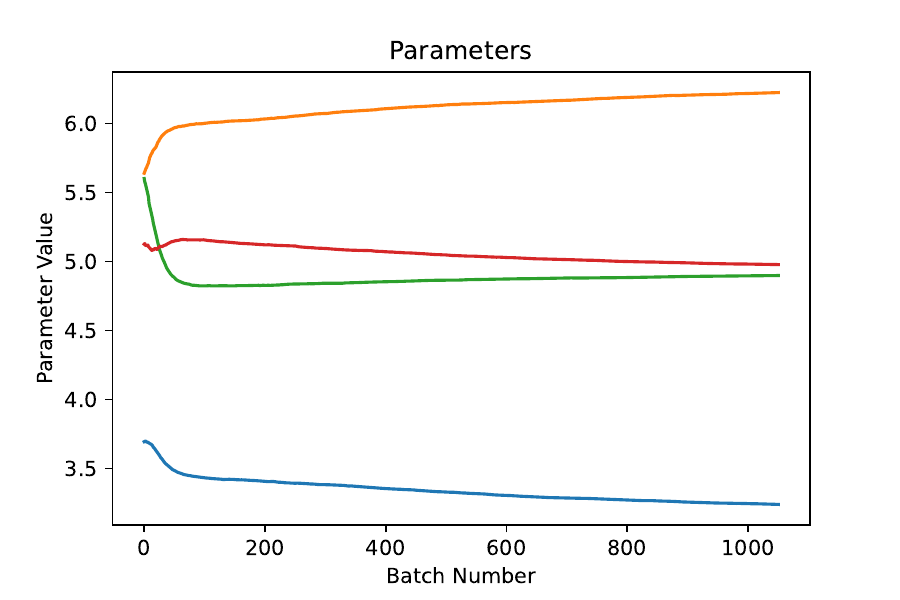}
\caption{Parameter curves for $C_2\times C_2$ symmetry.}
\label{fig:c2c2params}
\end{figure}

\begin{figure}[H]
\centering
\includegraphics[
width=3.5in
]{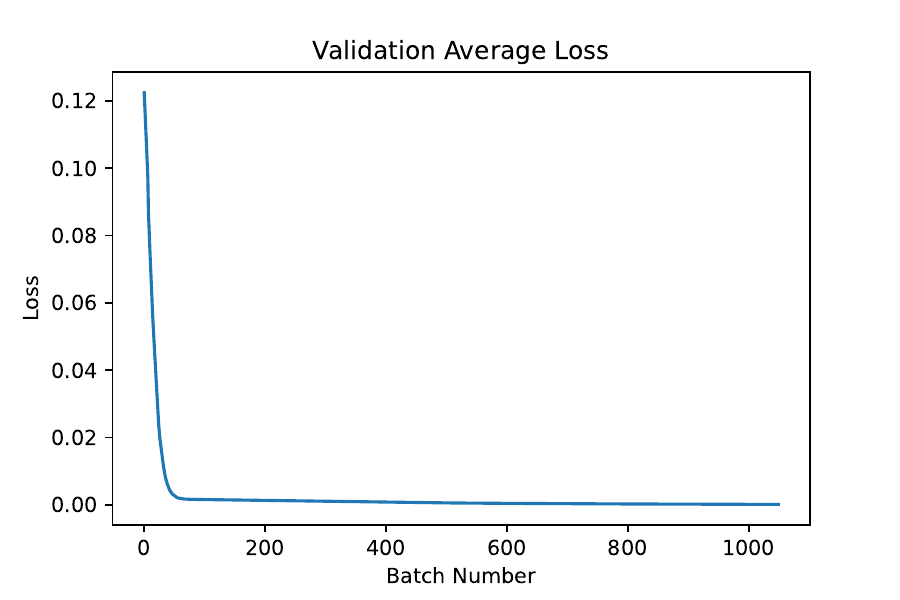}
\caption{Average loss for $C_2\times C_2$ symmetry.}
\label{fig:c2c2loss}
\end{figure}

Using the same reasoning as for the $C_2$ case, we expect to learn parameters which are equivalent up to some possible phase differences which cancel overall. This is indeed what we find, as the model learns the parameters $\theta_0 \approx 3.24\approx\pi$, $\theta_1\approx6.22\approx2\pi$,  $\theta_2\approx4.89\approx3\pi/2$, and $\theta_3\approx4.97\approx3\pi/2$. The latter two parameters are roughly equivalent, while the former two parameters are about $\pi/2$ from the latter parameter values in opposite directions, thereby canceling signs. See Figure~\ref{fig:c2c2_circ} for the circuit, Figure~\ref{fig:c2c2params} for the parameter curves, and Figure~\ref{fig:c2c2loss} for a plot of the average validation loss.

\subsection*{Example: $D_4$ Symmetry} 

Consider a $3\times3$ checkerboard distribution of points in the plane. This distribution demonstrates symmetry with respect to the dihedral group $D_4$ of order 8. In Figure~\ref{fig:checkerboard}, we plot several sample points which are separated into classes in alternating squares, forming just such a grid. Observe that a rotation $r$ by 90 degrees leaves this classification invariant, as does a flip $F$ about the diagonal line $y=x$. These symmetries together generate $D_4$, and so we take this to be our first representation.

\begin{figure}[H]
\centering
\includegraphics[
width=3.5in
]{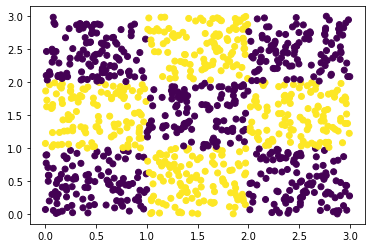}
\caption{Scatter plot of several data points in the checkerboard problem. Points having the same color have the same label. }
\label{fig:checkerboard}
\end{figure}

For the second representation, we will take the standard representation on four qubits, which is obtained by first identifying $D_4$ with the subgroup of the symmetric group $S_4$ generated by the permutations $(1\ 2\ 3\ 4)$ and $(1\ 2)(3\ 4)$ and then identifying these permutations with the corresponding permutation operators which act on four qubits. Thus, for the second representation we take $W(r)=\textnormal{SWAP}_{23}\textnormal{SWAP}_{12}\textnormal{SWAP}_{01}$ and $W(F)=\textnormal{SWAP}_{01}\textnormal{SWAP}_{23}$. Some operators which commute with these generators and are therefore $D_4$-invariant include a tensor product of single qubit operators $U^{\otimes4}$ as well as the permutation $\textnormal{SWAP}_{02}\textnormal{SWAP}_{13}$. Let us therefore apply these operations with $U=Y$ and measure with respect to $Z^{\otimes 4}$. For the architecture of the embedding, we apply $R_Z(\theta_{i+4}+x_{(i+1)\!\!\mod2})R_Y(\theta_i + x_{i\!\!\mod2})$ to the $i$-th qubit. See Figure~\ref{fig:d4_circ} for the circuit.

\begin{figure}[H]
\centering
\includegraphics[
width=3.5in
]{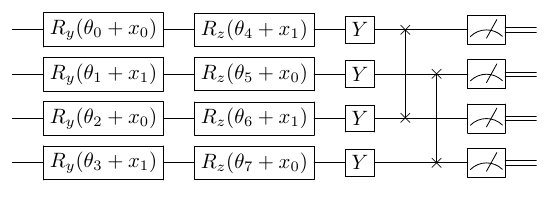}
\caption{Circuit for learning equivariance between two representations of $D_4$ using both $R_y$ and $R_z$ rotation gates. The parameters are randomly initialized and the data is encoded as $x=(x_0,x_1)$.}
\label{fig:d4_circ}
\end{figure}

\begin{figure}[h]
\centering
\includegraphics[
width=3.5in
]{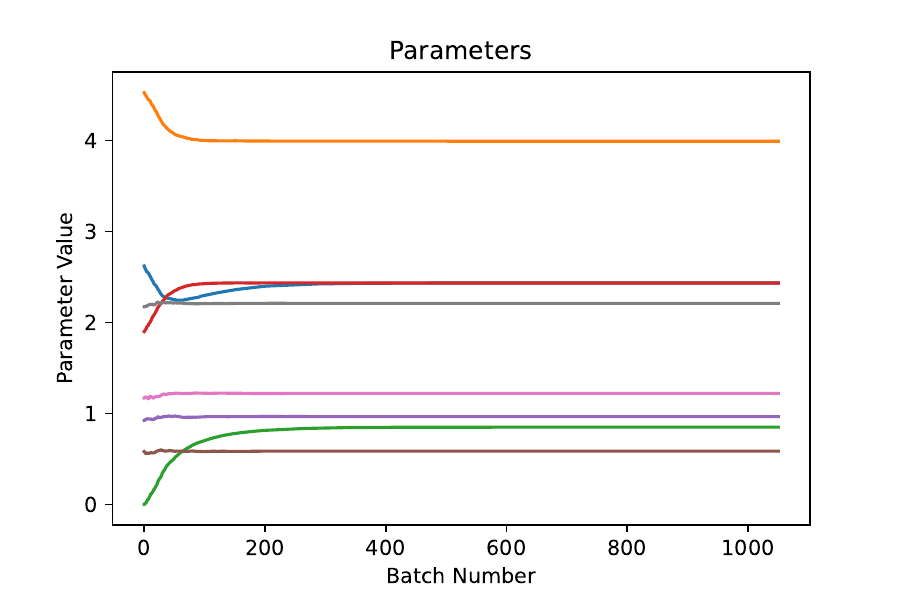}
\caption{Parameter curves for $D_4$ symmetry.}
\label{fig:d4params}
\end{figure}

The validation loss and parameter curves are shown in Figures~\ref{fig:d4loss} and \ref{fig:d4params}, respectively. 
The resulting embedding after training is nearly equivariant. Indeed, taking the data point $(0.40, 0.54)$, its rotation, and its reflection, and running the circuit with the trained parameters produces the nearly identical expected values $-0.05429996$, $-0.05429986$, and $-0.05430002$, respectively.

\begin{figure}[h]
\centering
\includegraphics[
width=3.5in
]{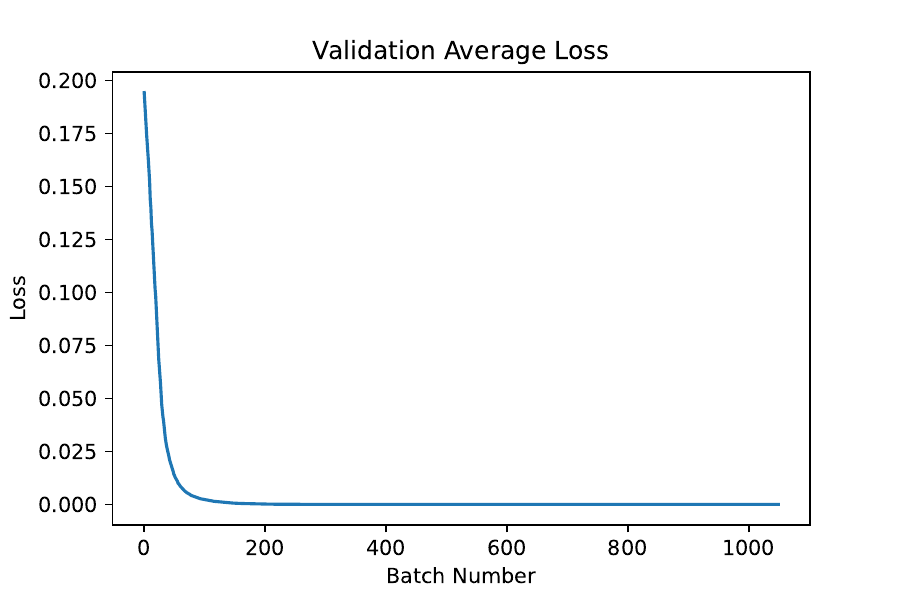}
\caption{Average loss for $D_4$ symmetry.}
\label{fig:d4loss}
\end{figure}

\begin{figure}[b]
\centering
\includegraphics[
width=3.5in
]{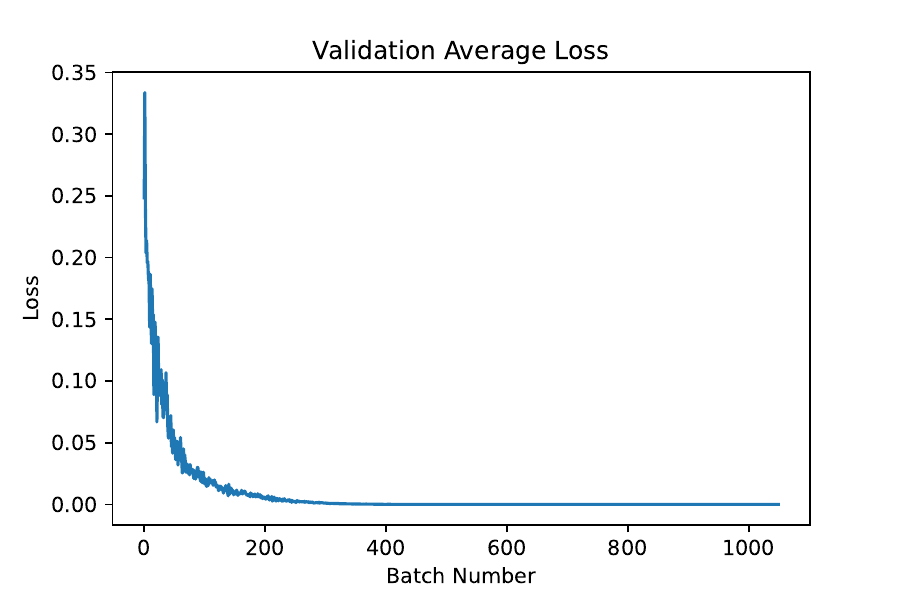}
\caption{Average loss for $S_6$ symmetry.}
\label{fig:s6loss}
\end{figure}

\begin{figure}[b]
\centering
\includegraphics[
width=3.5in
]{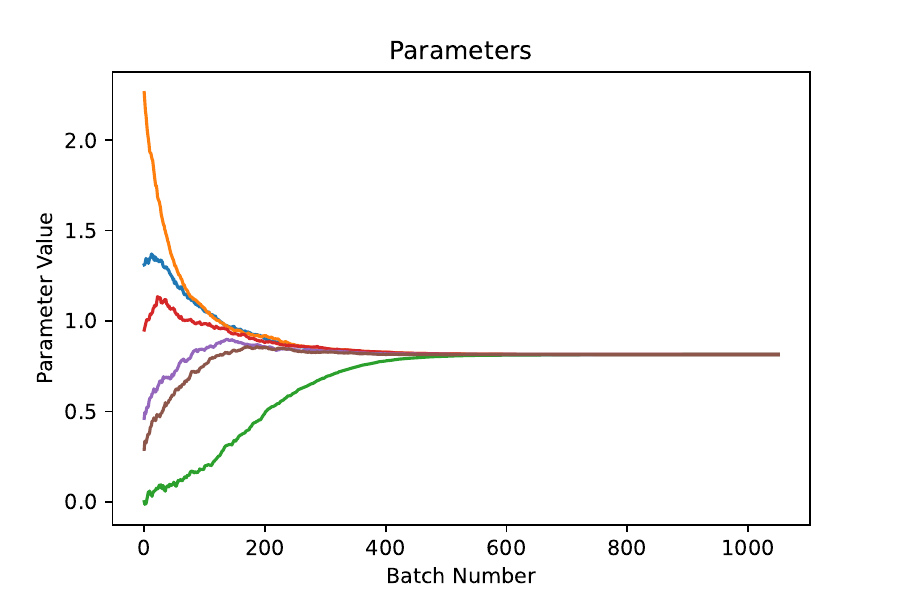}
\caption{Parameter curves for $S_6$ symmetry.}
\label{fig:s6params}
\end{figure}

\subsection*{Example: $S_6$ Symmetry}
Consider the symmetric group $S_6$ on 6 letters. A problem which possesses this symmetry is the classification of graphs with 6 nodes as being either connected or disconnected. A \textit{graph} is a collection of \textit{vertices} $V$ and \textit{edges} $E\subset V\times V$. When $(v,w)\in E$ if and only if $(w,v)\in E$ holds for every $v,w\in V$, we say that the graph is \textit{undirected}. Otherwise, the graph is called \textit{directed} or a \textit{digraph}. We say that a given graph is \textit{connected} if it has the property that for any vertices $v,w$, there is a sequence of edges $(v_i,w_i)$ with $i=0,1,\ldots,n$ such that $v=v_0$, $w_n=w$, and $w_i=v_{i+1}$ for all $i=0,\ldots,n-1$. Intuitively, the graph is connected if there is a path connecting any two vertices. Notice that the labeling of the vertices is arbitrary. We may choose to relabel the vertices with a permutation $\varphi:V\to V$ so long as vertices which are connected by an edge before the permutation remain connected after the permutation (equivalently, the edges are relabeled in precisely the same way). That is, we require that $\varphi$ is a bijection with the property that $(\varphi(v),\varphi(w))$ is an edge if and only if $(v,w)$ is an edge. Such a map is called a \textit{graph isomorphism}, and it can be shown that it preserves the property of connectedness. 

Since permutations of the vertices are equivalent to graph isomorphisms, we take the collection of such maps as our representation of $S_6$ in the data space. In the Hilbert space, we take the standard representation of $S_6$ given by permuting tensor products according to the given permutation in $S_6$. Since this group is generated by the elements $(1\ 2)$ and $(1\ 2\ 3\ 4\ 5\ 6)$, we see that the representation in Hilbert space is given by mapping $(1\ 2)$ to $\textnormal{SWAP}_{01}$ and mapping $(1\ 2\ 3\ 4\ 5\ 6)$ to $\textnormal{SWAP}_{01}\textnormal{SWAP}_{12}\textnormal{SWAP}_{23}\textnormal{SWAP}_{34}\textnormal{SWAP}_{45}$. Certainly a tensor power $U^{\otimes6}$ of any single qubit operator $U$ commutes with these generators and is therefore $S_6$-invariant. For this reason, we take $Z^{\otimes6}$ as our observable and apply the operation $U =Y^{\otimes6}$ 
after the embedding, which we give the following simple architecture: if $(i,j)$ is an edge in the graph with $i\ne j$, apply a controlled rotation with angle $\theta_i$ controlled off of qubit $i$ and targeting qubit $j$, and if $i=j$, then simply apply $R_Y(\theta_i)$ to qubit $i$. The circuit is shown in Figure~\ref{fig:s6_circ} for a particular graph defined by the notation \{0: [3, 0, 5, 4],
 1: [1, 5, 3, 0, 2],
 2: [5, 3, 2, 4],
 3: [5],
 4: [4, 2],
 5: [2, 5, 1, 3]\},
which indicates, for example, that there is an edge between vertex 0 and vertices 3, 0, 5, and 4.

The validation loss and parameter curves after training are shown in Figures~\ref{fig:s6loss} and \ref{fig:s6params}, respectively. The parameters converged to the same value, which is what we expect since any relabeling of the vertices (maintaining adjacency) gives rise to an equivalent graph.

\begin{figure*}[htp]
\centering
\includegraphics[
width=\textwidth
]{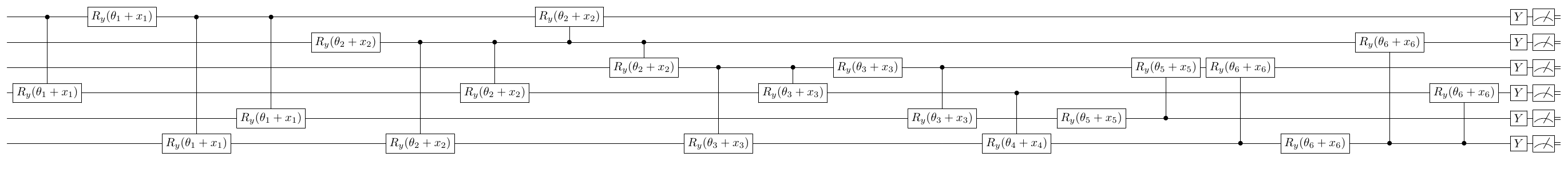}
\caption{Circuit for learning equivariance between two representations of $S_6$. In this figure, the data is taken to be the graph \{0: [3, 0, 5, 4],
 1: [1, 5, 3, 0, 2],
 2: [5, 3, 2, 4],
 3: [5],
 4: [4, 2],
 5: [2, 5, 1, 3]\} and the parameters are randomly initialized.}
\label{fig:s6_circ}
\end{figure*}

\section{Learning Intertwining Maps}\label{sec:intertwine}

We call a linear equivariant map between two representations an intertwining map. Such maps preserve both the linear vector space structure and the symmetry structure afforded by the representation. When the intertwining map is also a bijection, we call it an isomorphism of representations, and the representations are called equivalent. In representation theory, we seek to classify all irreducible representations according to this equivalence relation.

A limitation of the approach in Section~\ref{sec:equivariant} is that an embedding of the data space into Hilbert space is necessarily nonlinear, and so such an equivariant map cannot transfer the linear structure of $\mathcal{R}$ onto $\mathcal{H}$. Indeed, if the embedding were linear, then $\mathcal{E}(2x)=2\mathcal{E}(x)$. But if $\mathcal{E}(x)\in\mathcal{U}(\mathcal{H})$, then $2\mathcal{E}(x)\not\in\mathcal{U}(\mathcal{H})$, a contradiction. Fortunately, there is a way around this issue.

Suppose that we are given a dataset with some symmetry described by the representation $V:G\to GL(\mathcal{R})$. Let us denote by $\mathcal{E}:\mathcal{R}\to\mathcal{U}(\mathcal{H})$ the amplitude embedding, so that $T:\mathcal{R}\to\mathcal{H}$ defined by $T(x)=\mathcal{E}(x)\ket{\vec{0}}$ is the map sending each data element to its corresponding quantum state. While the amplitude embedding $\mathcal{E}$ is not linear, its unnormalized counterpart is. Indeed, the map $\Hat{\mathcal{E}}:\mathcal{R}\to\mathcal{U}(\mathcal{H})$ defined by $\Hat{\mathcal{E}}(x)=\|x\|\mathcal{E}(x)$ produces the linear operator which performs an unnormalized amplitude embedding, and so the Hilbert space level linear operator that we are looking for is $\Hat{T}:\mathcal{R}\to\mathcal{H}$ defined by $\Hat{T}=\Hat{\mathcal{E}}(x)\ket{\vec{0}}$. The trouble is that $\Hat{\mathcal{E}}(x)$ is not a unitary operator, and so we cannot implement it in a straightforward way in our quantum circuit. Instead, we resort to the following proposition linking the \textit{linear} equivariance of $\hat{T}$ to the equivariance of $T$.

\begin{lemma}\label{lem:amplitude}
    Let $T(x)=\mathcal{E}(x)\ket{\vec{0}}$ where $\mathcal{E}$ is the normalized amplitude embedding, and let $\Hat{T}(x)=\Hat{\mathcal{E}}(x)\ket{\vec{0}}$ where $\Hat{\mathcal{E}}$ is the unnormalized amplitude embedding. Then $\Hat{T}$ is an intertwining map between unitary representations $V$ and $W$ if and only if $T$ is equivariant between the same representations.
\end{lemma}
\begin{proof}
    It is clear that $\Hat{T}$ is linear as it maps the vector $x_0e_0+\cdots+x_ne_n\in\mathcal{R}$ to $x_0\ket{0}+\cdots+x_n\ket{n}\in\mathcal{H}$. Let $V$ be the representation on $\mathcal{R}$ and let $W$ be the representation on $\mathcal{H}$. Then for all $g\in G$, we have
    \begin{equation}
        W(g)\Hat{T}(x)=\Hat{T}(V(g)x)
    \end{equation}
    if and only if
    \begin{equation}
        \frac{1}{\|x\|}W(g)\Hat{T}(x)=\frac{1}{\|x\|}\Hat{T}(V(g)x),
    \end{equation}
    which is equivalent to
    \begin{equation}
        W(g)T(x)=\frac{\|V(g)x\|}{\|x\|}T(V(g)x).
    \end{equation}
    Thus, if $V$ is unitary, then the assumption that $\Hat{T}$ is an intertwining map is equivalent to the condition
    \begin{equation}
        W(g)T(x)=T(V(g)x),
    \end{equation}
    which is precisely the condition that $T$ is an equivariant map.
\end{proof}

Lemma~\ref{lem:amplitude} assumes that the embedding is the usual unparameterized amplitude embedding, and so its linear extension cannot learn an intertwining map between representations. Instead, we insert a parameterized unitary operator which we will use to learn equivariance between the representation induced by $\Hat{T}$ and another representation $W$ acting on the same Hilbert space. Let us denote by $A:G\to\mathcal{U}(\mathcal{H})$ the representation induced by $\Hat{T}$, which we recall satisfies
\begin{equation}
    \Hat{T}\circ V(g) = A(g)\circ\Hat{T}.
\end{equation}

We begin by performing the amplitude embedding on our data, after which we apply a parameterized unitary operator $U_\theta$. We then perform an operator $U_{inv}$ which is $G$-invariant with respect to the $W$ representation and measure with respect to an observable $\mathcal{O}$ that is likewise $G$-invariant. Let us again denote by $h_\theta$ the expected value of this circuit. We now claim that when $U_\theta$ is an equivariant map between $A$ and $W$, the estimate $h_\theta$ is $G$-invariant with respect to $V$. 

\begin{proposition}
    Let $V:G\to\mathcal{U}(\mathcal{R})$ and $W:G\to\mathcal{U}(\mathcal{H})$ be representations. The estimate
    \begin{equation}
        h_\theta(x)=\tr\left[U_{inv}U_\theta T(x)(U_{inv}U_\theta T(x))^\dagger\mathcal{O}\right]
    \end{equation}
    is $G$-invariant (i.e. $h_\theta(V(g)x)=h_\theta(x)$) whenever the unitary $U_{inv}$ and the observable $\mathcal{O}$ are $G$-invariant with respect to $W$ and $U_\theta$ is an intertwining map between $W$ and the representation $A$ induced by $\Hat{T}$.
\end{proposition}
\begin{proof}
Indeed, we have
\begin{equation}
    h_\theta(V(g)x)=\tr\left[U_{inv}U_\theta T(V(g)x)\left(U_{inv}U_\theta T(V(g)x)\right)^\dagger\mathcal{O}\right].
\end{equation}
Since $A$ is induced by $\Hat{T}$, we have by Lemma~\ref{lem:amplitude} that $T$ is an equivariant map between $V$ and $A$, and it follows that
\begin{equation}
    =\tr\left[U_{inv}U_\theta A(g)T(x)\left(U_{inv}U_\theta A(g)T(x)\right)^\dagger\mathcal{O}\right].
\end{equation}
Now applying the equivariance of $U_\theta$ produces
\begin{equation}
    =\tr\left[U_{inv}W(g)U_\theta T(x)\left(U_{inv}W(g)U_\theta T(x)\right)^\dagger\mathcal{O}\right].
\end{equation}
The $G$-invariance of $U_{inv}$ then gives
\begin{align}
    &=\tr\left[W(g)U_{inv}U_\theta T(x)\left(W(g)U_{inv}U_\theta T(x)\right)^\dagger\mathcal{O}\right]\\
    &=\tr\left[W(g)U_{inv}U_\theta T(x)(T(x))^\dagger U_\theta^\dagger U_{inv}^\dagger (W(g))^\dagger\mathcal{O}\right]\\
    &=\tr\left[U_{inv}U_\theta T(x)(T(x))^\dagger U_\theta^\dagger U_{inv}^\dagger (W(g))^\dagger\mathcal{O}W(g)\right],
\end{align}
where in the last equality, we have applied the cyclicity of the trace. Now applying the $G$-invariance of the observable $\mathcal{O}$ produces
\begin{align}
    &=\tr\left[U_{inv}U_\theta T(x)(T(x))^\dagger U_\theta^\dagger U_{inv}^\dagger (W(g))^\dagger W(g)\mathcal{O}\right]\\
    &=\tr\left[U_{inv}U_\theta T(x)(T(x))^\dagger U_\theta^\dagger U_{inv}^\dagger\mathcal{O}\right]\\
    &=\tr\left[U_{inv}U_\theta T(x)(U_{inv}U_\theta T(x))^\dagger\mathcal{O}\right]\\
    &=h_\theta(x).
\end{align}
This completes the proof.
\end{proof}

To train the equivariance of $U_\theta$, we sample from both the data space and the group and minimize the loss function
\begin{equation}
    L_g(x)=( h_\theta(V(g)x)-h_\theta(x)) ^2
\end{equation}
just as in Section~\ref{sec:equivariant}, and again optimize using one of the many available optimization routines. Thus, we have a procedure for learning an intertwining map between the representation $A$ induced by $\Hat{T}$ and the given representation $W$. 


\section{Conclusion}\label{sec:conclusion}

We have presented a novel approach to learning equivariant maps between representations of finite groups using variational quantum circuits. By parameterizing the embedding of the data space into Hilbert space, and training it to respect the inherent symmetries of the system, we have demonstrated a method for capturing and preserving group symmetries in quantum computations. Our exploration of various finite groups highlights the versatility and effectiveness of this approach, and each example showcases the capacity of variational quantum circuits to learn and respect different symmetries.

It should be noted that the choice of architecture in the embedding can have a profound influence on the ability of a model to learn equivariant maps. Indeed, an architecture which is too limited in scope may be incapable of representing an equivariant map between two particular representations. This can happen if there are too few parameters in the model or the gate set used is not expressive enough. The trivial example is a standard embedding procedure, which is not parameterized at all and therefore cannot learn an equivariant map. On the other hand, an architecture that is too broad in scope might hinder the ability of the model to train. This can happen, for example, if there are too many parameters in the model or the gate set is overly expressive. A well-designed architecture strikes a balance between expressiveness and trainability, avoiding limitations due to insufficient parameters while mitigating the challenges of over-parameterization.

Our findings contribute to the field of geometric quantum machine learning by showing that an important aspect of this theory, the embedding, can itself be learned. Moreover, this work provides a method for measuring the similarity between two representations, paving the way for further advancements in symmetry-based quantum information processing. These results demonstrate the potential for variational quantum circuits to be powerful tools for symmetry-respecting tasks.
\section*{Acknowledgements}
MLL, ZPB, and ENE acknowledge support from the DoD Smart Scholarship and OUSD(R\&E) SMART SEED grant. 

This document has been approved for public release, Distribution A. Distribution is unlimited.

\bibliographystyle{unsrt}
\bibliography{main}

\end{document}